\documentclass[10pt,conference]{IEEEtran}

\pdfoutput=1 

\usepackage[hidelinks]{hyperref}
\usepackage{amsmath}
\usepackage{amssymb}
\usepackage{mathtools}
\usepackage{physics}
\usepackage{relsize}
\usepackage{amsbsy}
\usepackage{cleveref}

\RequirePackage{tcolorbox}
\RequirePackage{booktabs}
\RequirePackage{braket}

\newtheorem{theorem}{Theorem}

\newtheorem{lemma}{Lemma}

\newtheorem{definition}{Definition}

\newtheorem{remark}{Remark}
\newenvironment{proof}{\textit{Proof.}}{\hfill$\square$}
\newcommand{\Sout}{S_{\mathsf{out}}}
\newcommand{\Sin}{S_{\mathsf{in}}}
\newcommand{\Dout}{D_{\mathsf{out}}}
\newcommand{\Din}{D_{\mathsf{in}}}
\newcommand{\psioutprime}{\ket{\psi_{\mathsf{out}}'}}

\newcommand{\vac}{\ket{\mathsf{vac}}}

\begin{document}

\title{Simple rules for two-photon state preparation with linear optics}

\author{
\IEEEauthorblockN{Grégoire de Gliniasty \textsuperscript{\textsection}}
\IEEEauthorblockA{Quandela - Massy, France\\
LIP6 - Paris, France}
\and
\IEEEauthorblockN{Paul Bagourd}
\IEEEauthorblockA{Quandela - Massy, France\\ EPFL - Lausanne, Switzerland}
\and
\IEEEauthorblockN{Sébastien Draux}
\IEEEauthorblockA{Quandela - Massy, France}
\and
\IEEEauthorblockN{Boris Bourdoncle}
\IEEEauthorblockA{Quandela - Massy, France}
}

\maketitle
\begingroup\renewcommand\thefootnote{\textsection}\footnotetext{gregoire.de-gliniasty@quandela.com}

\begin{abstract}
Entangling photons is a critical challenge for photonic quantum information processing: entanglement is a crucial resource for quantum communication and computation but can only be performed in a probabilistic manner when using linear optics.
In this work, we leverage a two-photon state matrix representation to derive necessary and sufficient conditions on two-photon entangling operations with linear optics.
We give a characterization of the input photonic states that can be used to prepare arbitrary two-qudit states in d-rail encoding with post-selection. 
We determine how many auxiliary photons are required to prepare any two-photon state with heralding.
In addition, we present a construction for generalized post-selected n-qubit control-rotation gates $ C^{n-1}Z({\varphi})$. 
\end{abstract}

\section{Introduction}

Photonics is ubiquitous in quantum information processing~\cite{FSS2018Photonic}: it is an ideal platform for quantum communication~\cite{LCS2023Recent}, computational advantage demonstration~\cite{ZDC2020Quantum, MLA2022Quantum}, near-term computation~\cite{maring2024versatile} and fault-tolerant computation \cite{de2023spin,bartolucci2023fusion}. However, computations are limited by the intrinsic non-determinism of entangling operations with linear optics, and finding reliable linear-optical entangling operations~\cite{knill2001scheme} is an important challenge for the field. Entanglement with linear optics can be achieved via measurement-induced non-linearity, as is done with post-selected controlled gates~\cite{ralph2002linear,fiuravsek2006linear}, fusion gates \cite{browne2005resource} and heralded gates such as Knill's CNOT \cite{knill2002quantum}. Recent research investigated the success probability for these operations by several means. For instance, the optimal success probability for the preparation or the measurement of Bell pairs and for heralded controlled gates have been investigated numerically \cite{fldzhyan2021compact,olivo2018ancilla}, and adaptive strategies, which consist in changing on the fly the circuit depending on some intermediate measurements, have been introduced to increase the success probability~\cite{bartolucci2021creation,hilaire2023linear}. Determining whether some classes of multidimensional entangled photonic states can be prepared at all, starting from a given separable multi-photon state and using linear optics, is equally important \cite{erhard2020advances, paesani2021scheme}, as such resource states 
have applications in device-independent cryptography~\cite{VPB2010Closing}, quantum key distribution~\cite{ICC2017Provably} or fault-tolerant quantum computation~\cite{Campbell2014Enhanced}. 

In this article, we analyze the feasibility of two-photon entangled state preparation using static interferometers. We make use of a convenient matrix representation of the states \cite{Calsamiglia2002Generalized,Lutkenhau1999Bell,kieling2008linear}. 
First, this enables us to characterize any state by the rank of its associated matrix, which surprisingly captures well the feasibility of the state preparation. Then, it allows us to reduce the linear optical transformation to matrix multiplication. In particular, it imposes certain conditions on the coefficients of the matrices of the states for a heralded or post-selected preparation to be achievable. 

We introduce the notations and this representation in Section \ref{sec: Prel&Not}.
Using this method, we derive a result for post-selected state preparation and a result for heralded state preparation. 
In Section \ref{sec: PS}, we present Theorem \ref{Th:PS}, which states that a two-qudit state can be prepared with post-selection from an input two-photon state if and only if the rank of the input state is at least half the rank of the two-qudit state. 
Then, in Section \ref{sec: H}, we present Theorem \ref{Th:H}, which states that a two-photon state can be prepared with heralding if and only if the number of input photons is greater than or equal to the state's rank.
In addition, we present a construction for a generalized control-rotation post-selected gate in Section \ref{sec: CZs}.

\section{Preliminaries and notations} \label{sec: Prel&Not}

\subsection{Linear optics}

In the following and unless stated otherwise, $n\in\mathbb{N}$ denotes the total number of photons and $m\in\mathbb{N}$ the total number of modes, $a_i^\dagger, a_i$ are the creation and annihilation operators on mode $i$, and $\vac = \ket{0}^{\otimes m }$ is the vacuum state. 

States with $n$ photons in $m$ modes are described by the number of photons in each mode as 
\begin{equation}
\begin{array}{ll}
    \ket{n_1\ldots,n_m} = \prod\limits_{i=1}^m \frac{(a_i^\dagger)^{n_i}}{\sqrt{n_i!}} \vac, \text{ with } \sum_{i=1}^m n_i = n. 
\end{array}
\end{equation}
We call them Fock states, and we denote their Hilbert space by $\mathcal{F}_n^m$. Any state $\ket{\psi}\in \mathcal{F}_n^m$ can be described with a complex homogeneous multivariate polynomial of degree $n$, $P_{\psi}$:
\begin{equation}
    \begin{array}{ccl}
    \ket{\psi} & = &  P_\psi(a_1^\dagger,\ldots,a_m^\dagger)\vac .
    \end{array}
\end{equation}
For brevity we sometimes write $\mathbf{a^\dagger}=(a_1^\dagger,\ldots,a_m^\dagger)$, so $ \ket{\psi} = P_\psi(\mathbf{a^\dagger})\vac$.

Linear optics allows any unitary transformation on creation operators. A unitary transformation acting on $m$ modes can be described with a matrix $U\in \mathcal{SU}(m)$ acting on the creation operators $a_i^\dagger$ as 
\begin{equation}
\label{eq:unitary}
U: a^\dagger_j \mapsto U(a^\dagger_j) = \sum\limits_{i=1}^m u_{ij}a^\dagger_i .
\end{equation}
In that case the linear-optical effect on a multi-photon state is:
\begin{equation}
\begin{array}{ccl}
    \mathcal{U} : &  \ket{\psi}=\ket{n_1,\ldots,n_m} & = \prod\limits_{i=1}^m \frac{(a_i^\dagger)^{n_i}}{\sqrt{n_i!}}\vac   \\
     & \mapsto  & = \prod\limits_{i=1}^m \frac{(U(a_i^\dagger))^{n_i}}{\sqrt{n_i!}}\vac   \\
       &   & = P_\psi(U(\mathbf{a^\dagger}))\vac.
\end{array}
\end{equation}
We use the notation $\mathcal{U}$ to denote the unitary transformation on the Fock space induced by the linear-optical unitary $U$.

We can describe the coefficients of a Fock state's polynomial after any unitary transformation with the permanent of a submatrix \cite{scheel2004permanents}. For two $n$-photon Fock state  $\ket{\mathbf{k}} = \ket{k_1,\ldots,k_m}$ and $\ket{\boldsymbol\ell}=\ket{\ell_1,\ldots,\ell_m}$, 
\begin{equation}
    \label{eq:Per}
    \bra{\mathbf{k}}\mathcal{U}\ket{\boldsymbol\ell}= \frac{\text{Per}(U_{\mathbf{k},\boldsymbol\ell})}{\sqrt{k_1!\ldots k_m!\ell_1!\ldots \ell_m!}},
\end{equation}
where $U_{\mathbf{k},\boldsymbol\ell}$ is an $n\times n$ matrix obtained by repeating every row $i$ of $U$ $k_i$ times, then repeating every column $j$ of this matrix $\ell_j$ times. 

\subsection{Two-photon states}
We use a convenient representation of two-photon Fock states that relies on the fact that the state's polynomial is quadratic, and can thus be represented as a symmetric $m\times m$ matrix  $S$ \cite{Calsamiglia2002Generalized,Lutkenhau1999Bell,kieling2008linear}:
 \begin{equation}
 \begin{array}{ccl}
     P(\mathbf{a^\dagger})\vac & = & (\mathbf{a^\dagger})^{T}S\mathbf{a^\dagger} \vac \\
                               & = & \sum\limits_{i,j=1}^m S_{ij} a_i^\dagger a_j^\dagger\vac.
\end{array}
\end{equation}
We alternatively refer to a two-photon state by its matrix $S$ or by its state $\ket{\psi}$. Unitary transformations act on $S$ as: 
 \begin{equation}
 \begin{array}{ccr}
    P(U(\mathbf{a^\dagger}))\vac&=& (U(\mathbf{a^\dagger}))^{T}SU(\mathbf{a^\dagger}) \vac\\
                                &=& (\mathbf{a^\dagger})^{T}U^{T}SU\mathbf{a^\dagger} \vac.
    \end{array}
 \end{equation}
 
A useful tool for two-photon states is Takagi's factorization, which states that for any square complex and symmetric matrix $S$, there exist a unitary matrix $V$ and a diagonal non-negative matrix $D$ such that  $D = V^T S V$. The diagonal elements are the non-negative square roots of the eigenvalues of $S^TS$.

From the diagonal representation one can see that the rank of $S$ is an invariant of the state under linear optics. We define the rank of the state as the rank of its matrix $\text{rank}(S)$, which is equal to the number of non-zero diagonal elements in $D$. Moreover, the normalization of the state $S$ imposes $\text{Tr}(D^2) = \frac{1}{2}$.

\subsection{Unitary extension and auxiliary modes}
The mode transformation $A$ that we derive might not be unitary in the first place, in which case it would not be implementable with linear optics. However, embedding $A$ in a unitary matrix is always possible at the cost of adding auxiliary modes, which correspond to additional modes in which we input no photons and impose to detect no photons by post-selection. A unitary embedding with auxiliary modes does not change the transformation $A$ on the relevant modes, however it can decrease the success probability.

The following lemma will be useful for the subsequent constructions.
\begin{lemma}[Unitary extension]\label{lemma:UnitExt}
    Let $A$ be an $m_1 \times m_2$ non-zero complex matrix with maximum singular value $\sigma_1>0$. Then there exists a unitary extension 
    \begin{equation}
        \begin{array}{cc}
           U = \begin{pmatrix} \frac{1}{\sigma_1} A & * \\ * & * \end{pmatrix} \in \mathcal{SU}(N) & \text{ where } N \leq m_1+m_2 .
        \end{array} \label{eq:unitExt}
    \end{equation}
\end{lemma}
\begin{proof}
    Without loss of generality, we assume that $m_1\leq m_2$.
    Let $ \frac{1}{\sigma_1} A = V_1\begin{pmatrix} D\\0\end{pmatrix}V_2^\dagger$ be a singular value decomposition of $ \frac{1}{\sigma_1}A$. $D$ is an $m_1 \times m_1$ diagonal matrix with non-negative real numbers on the diagonal.
    Let us define the two $(m_1+m_2) \times (m_1+m_2) $ unitary matrices 
    $$ \begin{array}{ccc}
        U_1 = \begin{pmatrix}  V_1 & 0\\ 0 & V_2\end{pmatrix}& \text{ and }&   U_2 = \begin{pmatrix}  V_2^\dagger & 0 \\ 0 & V_1^\dagger\end{pmatrix}
    \end{array}. $$
    Consider the matrix $\sqrt{I_{m_1}\mathsmaller{-}D^2}$ (all elements of $D$ are smaller than $1$ by construction). We then have the following embedding :
    $$\begin{pmatrix} \frac{1}{\sigma_1} A & *  \\ * & * \end{pmatrix} =
                         U_1
                        \begin{pmatrix}  D & 0 & \sqrt{I_{m_1}\mathsmaller{-}D^2}\\
                                         0 & I_{m_2-m_1} & 0\\
                                        \sqrt{I_{m_1}\mathsmaller{-}D^2} & 0 & -D\end{pmatrix}   
                        U_2,     $$
which is the product of three unitaries, hence unitary itself.
\end{proof}
\section{Post-selected two-qudit state preparation} \label{sec: PS}
In this section we introduce our result that relates the feasibility of the post-selected preparation of a 2-qudit state to the rank of the input state. 
Dual-rail is widely used for qubit encoding \cite{kok2010introduction} in photonics. 
This can be generalized to qudits with d-rail encoding. An arbitrary qudit state $\ket{\psi}$ is then encoded as a single creation operator over $d$ modes:
\begin{equation}
    \ket{\psi} = \sum_{i=0}^{d-1} \alpha_{i}\ket{i} =\sum_{i=0}^{d-1} \alpha_{i} a^\dagger_i\vac.
\end{equation} 
We encode two qudits of dimension $d_1$ and $d_2$ by encoding the first qudit on the first $d_1$ modes $\{a^\dagger_0,\ldots,a^\dagger_{d_1-1}\}$ and the second qudit on the $d_2$ following modes $\{a^\dagger_{d_1},\ldots,a^\dagger_{d_1+d_2-1}\}$.
Any such state can be described by a matrix $C$ of size $d_1\times d_2$ :
\begin{align}
    \ket{\psi} = & \sum\limits_{i=0}^{d_1-1} \sum\limits_{j=0}^{d_2-1} C_{ij}\ket{i}\otimes\ket{j} \nonumber \\ 
                  = &  \sum\limits_{i=0}^{d_1-1} \sum\limits_{j=0}^{d_2-1} C_{ij} a^\dagger_i a^\dagger_{d_1+j}\vac.
\end{align}
$\ket{\psi}$ is also a two-photon state that can be described by 
\begin{equation}
    S = \frac{1}{2}\begin{pmatrix}
        0   & C \\
        C^T & 0
    \end{pmatrix}.
\end{equation}
Note that $\text{rank}(S) = 2 \text{ rank}(C)$.
Post-selection in linear optical quantum computing corresponds to conditioning an experiment on the property that ``the measured state stays in the encoding'', as with the post-selected CNOT gate \cite{ralph2002linear} or the post-selected fusion gate \cite{browne2005resource}.
This allows the non-deterministic generation of entanglement with linear optics which is impossible deterministically.

In terms of two-qudit state preparation, this translates into conditioning the experiment upon measuring any two-photon state in the computational basis
$$\{a^\dagger_ia^\dagger_{d+j} : 0 \leq i \leq d_1\mathsmaller{-}1, d_1 \leq j \leq d_1+d_2\mathsmaller{-}1 \},$$
and rejecting any other term, as those with two photons in the same qudit's mode or those with photons in auxiliary modes.
This means that any state's matrix $S$ such that $S_{ij} = C_{ij}$ for $0 \leq i \leq d_1\mathsmaller{-}1, \text{ and } d_1 \leq j \leq d_1+d_2\mathsmaller{-}1$ is sufficient, and there are no constraints on the other coefficients, but the symmetry and normalization of $S$.

In particular, by adding $a$ auxiliary modes, preparing any $d_1+d_2+a$ modes state of whose matrix is of the form 
$$ S = \frac{N}{2}\begin{pmatrix}
                            *   & C & *\\
                            C^T & * & * \\
                            *   & * & *
                     \end{pmatrix}, $$
where $N$ is a normalization factor, is sufficient for preparing a two-qudit state, $C$, in post-selection.
We will note $*$ when there is no other constraints on the matrices' blocks than  the symmetry and normalization required to describe a two-photon state.

\begin{definition}
    Given a $m$-mode two-qudit output states $\Sout$ such that
    $$ \Sout = \frac{1}{2}\begin{pmatrix} 0 & C \\C^T & 0 \end{pmatrix}$$
    and a  $m$-mode two-photon state $\Sin$,
    we say it is possible to prepare $\Sout$ with post-selection from $\Sin$ if and only if there exist a number of auxiliary mode $a\in \mathbb{N}$, a unitary matrix $U \in \mathcal{SU}(m+a)$ and a non-zero coefficient $\alpha\in ]0,1] $ such that 
    \begin{equation}
        U^T \begin{pmatrix} \Sin & 0 \\ 0 & 0 \end{pmatrix} U = 
              \begin{pmatrix} \alpha \begin{pmatrix} * & C \\C^T & * \end{pmatrix} & * \\
                                      *    &            *                  &       \end{pmatrix}.\label{eq:defPS}
    \end{equation}
\end{definition}

\begin{theorem} \label{Th:PS}
    Let $\Sout$ describe a 2-qudit state and $\Sin$ a two-photon state.
    $\Sout$ can be prepared with post-selection if and only if $  \text{rank}(\Sout) \leq 2 \ \text{rank}(\Sin)$.
\end{theorem}

\begin{proof} We remind that with two-qudit of dimension $d_1$ and $d_2$, $C$ is a $d_1\times d_2$ matrix.
We assume without loss of generality that $d_1>d_2$.
We have that $\text{rank}(\Sout) = 2\ \text{rank}(C)$.

First, we suppose that $\Sout$ can indeed be prepared in post-selection.
We compare the rank of the matrices involved in \cref{eq:defPS}.
The matrices $U$ and $U^T$ are unitaries, so they preserve the rank.
Leveraging that the rank of a matrix is greater than the rank of any of its sub-matrices.
\begin{align*}
        \text{rank}(\Sin) & = \text{rank}(U^T\Sin U) \\
                         & \geq \text{rank}(C)\\
                         & \geq \text{rank}(\Sout)/2.
\end{align*}

Hence, 
$$ \text{rank}(\Sout) \leq 2\text{rank}(\Sin) .$$

Now suppose that $\text{rank}(\Sout) \leq 2\text{rank}(\Sin)$. We are going to exhibit a matrix $U$ that reproduces \cref{eq:defPS}.
Consider the following state : $$S_{ps} = \frac{N}{2}\begin{pmatrix}
        A   & C \\
        C^T & B
    \end{pmatrix},$$
for $A$ and $B$ symmetric matrices of the appropriate size and $N$ a positive normalization factor.
We show below that there exist symmetric matrices $A$ and $B$ such that $\text{rank}(S_{ps}) = \text{rank}(C)$.

Let $C = V_1\begin{pmatrix}D\\0\end{pmatrix}V_2^\dagger$ be a singular value decomposition of $C$.
$D$ is a non-negative diagonal matrix.

We set both $A$ and $B$ as
$$ \begin{array}{cc}
    A = V_1 \begin{pmatrix} D & 0\\ 0 & 0\end{pmatrix}V_1^T & B = (V_2^T)^\dagger D V_2^\dagger
\end{array}.$$
They are both symmetric by construction, hence $S_{ps}$ is symmetric and can represent a two-photon state.

Moreover, we have the following relation : 
$$ \begin{pmatrix} V_1^\dagger & 0 \\ 0 & (V_2^\dagger)^T \end{pmatrix}
  \begin{pmatrix} A & C \\ C^T & B \end{pmatrix}
  \begin{pmatrix} (V_1)^T & 0 \\ 0 & V_2^\dagger\end{pmatrix} 
  = \begin{pmatrix} D & 0 & D \\ 0 & 0 & 0 \\ D & 0 & D \end{pmatrix}
.$$
Hence there exist matrices symmetric matrices $A$ and $B$ such that $\text{rank}(S_{ps}) =\text{rank}(D) =\text{rank}(C)$.

Let $\Din$ and $D_{ps}$ be the non-negative diagonal matrices obtained from Takagi's factorization of matrices $\Sin$ and $S_{ps}$ such that for complex unitaries $U_i$ and $U_{ps}$: 
\begin{equation}
    \begin{array}{ccc}
         \Din = U_i^T \Sin U_i & \text{ and } & S_{ps} = U_{ps}^T D_{ps} U_{ps}.   
    \end{array}
\end{equation}
There exists a rescaling diagonal non-negative matrix $\Lambda$ such that $D_{ps} = \Lambda^T \Din \Lambda$ if and only if $$\text{rank}(D_{ps})\geq\text{rank}(D_{i}),$$
which is the case by assumption.

Hence, $S_{ps}= (U_i\Lambda U_{ps})^T \Sin (U_{i}\Lambda U_{ps}).$ 
This transformation is not necessarily unitary because of $\Lambda$, however with a unitary extension, this transformation is possible with linear optics at the cost of auxiliary modes. 
With the help of the unitary extension of \cref{lemma:UnitExt}, there exist $\alpha \in \mathbb{R}$ and $U\in\mathcal{SU}(N)$ with $N\leq 2(m_1+m_2)$, given by the unitary extension of $(U_{i}\Lambda U_{ps})$,  verifying \cref{eq:defPS}. 
\end{proof}

\begin{remark}
    The construction of the unitary matrices that allow the preparation in post-selection is detailed in the proof.
    The probability of success of such transformation can be derived from the unitary transformation.   
\end{remark}
\begin{remark}
    This theorem advocates for the use of qubit instead of qudits in linear optics.
    One can see that Bell pairs in qudits have rank $2d$ and two single photons have rank 2.
    When working in qubits, one can build a post-selected Bell pair with two single photons, as it take a rank-$2$ to a rank-$4$ state, verifying the $\times 2$ condition of \cref{Th:PS}.
    However for higher dimension qudits, $d>2$, the post-selected entangling gate would have to take a rank-$2$ state to a rank $>4$, which is not possible according to \cref{Th:PS}.
    One would need to prepare exactly a Bell pair in dimension $d$ in order to prepare in post-selection a Bell pair in dimension $2d$.
\end{remark}
  
\section{Heralded two-photon state preparation} \label{sec: H}
In this section we introduce our result that relates the feasibility of the heralded preparation of a two photon-state to the number of input photons.

\begin{definition} \label{def:H}
    Let $\Sout$ be a two-photon state defined over $m$ modes.
    We say that it is possible to prepare $\Sout$ with heralding from $n$ single photons if and only if there exist a number of auxiliary heralding modes $h\in\mathbb{N}$, a heralding signal $ \mathbf{s} = (s_1,\ldots, s_h) \in \mathbb{N}^{h}$ such that $\sum s_i = n\mathsmaller{-}2$, a number of auxiliary vacuum modes $a\in\mathbb{N} $ such that $n \leq m+h+a$, a success probability $p_s\in ]0,1]$ and a matrix $U\in\mathcal{SU}(m+h+a)$ verifying :
    \begin{equation}
    \begin{array}{c}
    \left({}_a\bra{0,\ldots}{}_h\bra{ \mathbf{s}}\otimes I_{m}\right) \mathcal{U}\ket{1,\ldots}_n\ket{0,\ldots}_{m+h+a-n} \\
    = (\mathbf{a^\dagger})^T\left( \sqrt{p_s}\Sout\right)\mathbf{a^\dagger}\vac.
    \end{array}
        \label{eq:defH}
    \end{equation}
\end{definition}

\begin{theorem}\label{Th:H}
    Let $\Sout$ be a two-photon state over $m$ modes.
    $\Sout$ can be prepared from $n$ single photons with heralding if and only if $n \geq \text{rank}(\Sout)$.
\end{theorem}

\begin{proof}
We take Takagi's factorization of $\Sout$, $\Dout = V^T \Sout V$ and we work on the diagonal state without loss of generality, as both states are equivalent by unitary transformation, here with $V \in \mathcal{SU}(m)$. 
Hence, we now aim to prepare the state 
\begin{equation}\label{eq:SoH}
    \begin{array}{ccc}
         \psioutprime & = & (\mathbf{a}^\dagger)^T V^T \Sout V \mathbf{a}^\dagger\vac.   \\
    \end{array}
\end{equation}

On the other hand, we have, from the definition, a unitary transformation $U\in\mathcal{SU}(m+h+a)$.
We can without loss of generality assume that the input photons are in the first $n$ modes, thus denoting $\ell_i$ the first $n$ elements of the rows of $U$:
$$\ell_i = (u_{i,0},\ldots,u_{i,n\mathsmaller{-}1}).$$
Moreover, for brevity, we condense the rows of the heralding modes, and we write
$$
\begin{array}{cccc}
    \text{Per}(\cdot,\cdot,{\boldsymbol \ell_{\mathbf{s}}}) = 
\text{Per}(\cdot,\cdot, & \underbrace{\ell_{m},\ldots,\ell_{m}}, & \ldots & \underbrace{\ldots,\ell_{m+h\mathsmaller{-}1}}) \\
     & s_1 \text{times} &     & s_h \text{times}
\end{array},$$
where $s_1,\ldots, s_h$ define the signal on the heralding modes.

Hence, using the definition and the unitary transformation rule on Fock states \cref{eq:Per}, we can write
\begin{equation} \label{eq:PerH}
    \begin{array}{r}
          \left({}_a\bra{0,\ldots}{}_h\bra{\mathbf{s}}\otimes I_{m}\right) \mathcal{U}\ket{1,\ldots}_n\ket{0,\ldots}_{m+h+a-n}       \\
          =  \sum\limits_{i=0}^{m-1}\sum\limits_{j=0}^{i} \frac{\text{Per}(\ell_i,\ell_i,{\boldsymbol\ell_{\mathbf{s}}})}{\sqrt{s_1!\ldots s_h!}}  \frac{a_i^\dagger a_j^\dagger}{\sqrt{2}^{\delta_{ij}}}\vac \\
          = \sum\limits_{i,j=0}^{m-1} \frac{\text{Per}(\ell_i,\ell_i,\boldsymbol\ell_{\mathbf{s}})}{2\sqrt{s_1!\ldots s_h!}}  \sqrt{2}^{\delta_{ij}}a_i^\dagger a_j^\dagger\vac.
    \end{array}
\end{equation}

First, we suppose that $\text{rank}(D)\leq n$.
Let us analyze the bi-linear form defined by $\text{Per}(\cdot,\cdot,\boldsymbol\ell_{\mathbf{s}}) : \mathbb{C}^n\times \mathbb{C}^n \mapsto \mathbb{C}$ for some rows $(\ell_m,\ldots,\ell_{m+h\mathsmaller{-}1}) \in (\mathbb{C}^n)^h$.
Consider its matrix representation denoted $F$. It is symmetric complex because of its construction from the permanent, and it is entirely determined by $\boldsymbol\ell_{\mathbf{s}}$. Takagi's factorization can be applied to it, resulting in a basis of  $\mathbb{C}^n$, $\{e_i, 0\leq i <n\}$   such that  $\text{Per}(e_i,e_j,\boldsymbol\ell_{\mathbf{s}})= \delta_{ij} F_{ii}$ for some non-negative values $F_{ii}$.
In particular, the rank of $F$ is at most $n$ because of its size, and it can be attained, for example by setting every line of  $\boldsymbol\ell_{\mathbf{s}}$ to $\sqrt{\frac{1}{n\mathsmaller{-}2}}(1,\ldots,1)$
\footnote{We easily exhibit $n$ linearly independent eigenvectors whose eigenvalues are non-zero, ensuring the rank is $n$. Take $v_n = (1,\ldots,1)$ , its eigenvalue is $m-1$. Take $v_i= (0\ldots,\underset{i, i+1}{\underbrace{1,-1}},\ldots)$ for $i<n$, their eigenvalues are $-1$.}.
We then take any $\boldsymbol\ell_{\mathbf{s}}$ such that $\text{rank}(F) =n$.

The first $n$ values of $F_{ii}$ are non-zero and since $\text{rank}(D)\leq n$ by assumption, one can define $m$ vectors, $\ell_i$ as 
$$ \begin{array}{cccc}
\ell_i& = & \sqrt{\sqrt{2s_1!\ldots s_{h}!}\frac{D_{ii}}{F_{ii}}} e_i,&\quad \forall i\in \{0,\ldots,n\mathsmaller{-}1\}\\
\ell_i & = & e_i,&\quad \forall i\in \{n,\ldots,m\mathsmaller{-}1\}
\end{array}.
 $$
This translates into  $$\text{Per}(\ell_i,\ell_j,\boldsymbol\ell_{\mathbf{s}})=\sqrt{2s_1!\ldots s_{h}!} D_{ii}\delta_{ij}.$$
Now we have well-defined rows $\ell_i$ that we embed in an unitary matrix using \cref{lemma:UnitExt}. 
There exist a number of auxiliary modes $a$ and a non-zero $\alpha$ such that we can construct an unitary $U\in\mathcal{SU}(m+h+a)$ of the form 
\begin{equation}
\begin{array}{cr}
    U = \begin{pmatrix}  \alpha A & *\\ *  & * \end{pmatrix}, & \text{  where }
    A = \begin{pmatrix}\ell_0 \\ \vdots \\ \ell_{m+h\mathsmaller{-}1}\end{pmatrix}. 
\end{array}
\end{equation}
In particular,
\begin{equation}
    \begin{array}{r}
          \left({}_a\bra{0,\ldots}{}_h\bra{\mathbf{s}}\otimes I_{m}\right) \mathcal{U}\ket{1,\ldots}_n\ket{0,\ldots}_{m+h+a\mathsmaller{-}n}       \\
          =  \sum\limits_{i=0}^{m-1}\sum\limits_{j=0}^{i} \frac{\alpha^n\text{Per}(\ell_i,\ell_j,\boldsymbol\ell_{\mathbf{s}})}{\sqrt{s_1!\ldots s_h!}}  \frac{a_i^\dagger a_j^\dagger}{\sqrt{2}^{\delta_{ij}}}\vac \\
          = \alpha^n\sum\limits_{i=0}^{m-1} D_{ii}  (a_i^\dagger)^2\vac .
    \end{array}
\end{equation}
Hence, if $\text{rank}(D)\leq n$, we can prepare the state $\Dout$ with heralding.

Now we suppose that the state $\Dout$ can be prepared with heralding.
Relating \cref{eq:SoH} and \cref{eq:PerH}, we can assert that  $\Dout$ can be prepared with heralding from $n$ single photons if and only if there exist some $a$, $h$, $\mathbf{s}$, $p_s$, $U$, according to \cref{def:H}, such that
\begin{equation*}
\forall i,j\in \{0,\ldots,m\mathsmaller{-}1\}, \ 
    \frac{\text{Per}(\ell_i,\ell_j,\boldsymbol\ell_{\mathbf{s}})\sqrt{2}^{\delta_{ij}}}{2\sqrt{s_1!\ldots s_{h}!}} = \sqrt{p_s}D_{ii}\delta_{ij}.
\end{equation*}
In particular here, 
\begin{equation*}
    \text{Per}(\ell_i,\ell_j,\boldsymbol\ell_{\mathbf{s}})= \sqrt{2s_1!\ldots s_{h}!p_s}D_{ii}\delta_{ij},
\end{equation*}
and $\text{rank}(D)\leq \text{rank}( \text{Per}(\cdot,\cdot,\boldsymbol\ell_{\mathbf{s}}))$. 
By  construction of  the bi-linear form $\text{Per}(\cdot,\cdot,\boldsymbol\ell_{\mathbf{s}}) : \mathbb{C}^n\times \mathbb{C}^n \to \mathbb{C}$ , its maximal rank is $n$. Hence $\text{rank}(D)\leq n$.
\end{proof}

\begin{remark}
    It has been conjectured that given two two-photon states $S_1$ and $S_2$, transforming $S_1$ into $S_2$ would require $\text{max}(0,\text{rank}(S_2)\mathsmaller{-}\text{rank}(S_1))$ single photons as resource \cite{kieling2008linear}.
    \Cref{Th:H} lends support to this conjecture, as we can interpret it as starting from a rank-$2$ state and transforming it into a rank-$n$ state with the help of $n\mathsmaller{-}2$ single photons.
\end{remark}
\begin{remark}
    Bell state measurements for qudits were introduced in~\cite{zhang2019quantum,luo2019quantum}, but they require auxiliary photons: for instance, for $d=3$, one additional photon is required.  
    This result can be formulated as an application of our \Cref{Th:H}, where $n\geq 3$ photons are required to obtain a state of rank $3$, followed by \Cref{Th:PS}, with a post-selection circuit that allows the non-deterministic measurement of a state of rank $2\times 3$. 
\end{remark}
\section{Controlled photonic gates} \label{sec: CZs}

In this section we propose a generalization to $n$ qubits of the post-selected controlled gates. 
The construction is similar to previous work \cite{KielingControlled}.
\definition[Generalized controlled-Z rotation]{Let $n$ be the number of qubit and $\varphi \in \left[0,2\pi\right]$. 
The n-qubit controlled gate is the unitary operation defined as 
\begin{equation}
    C^{n-1}Z({\varphi}) = Id_{2^n} + (e^{i\varphi}-1)\ket{1,\ldots,1}\bra{1,\ldots,1}.
\end{equation}
}

\subsection{Construction of the gate}
Let $a^{\dagger}_i$ and $b^{\dagger}_i$ the creation operators modes associated to the dual rail encoding of the $i^{th}$ qubit, with respectively the $b^\dagger$ for $\ket{0}$ and $a^\dagger$ for $\ket{1}$.
Let $\alpha$ be any root of $\alpha^n =  (e^{i\varphi}-1)$,  $I_n$ be the identity over $n$ modes, $J_n$ be any order $n$ cyclic permutation matrix over $n$ modes.
For simplicity, we will take the permutation that takes $i$ to $i+1 \text{ mod } n$. Also, we add a probability factor $p_s$ that will permit a unitary extension.
Consider the transformations over modes $a^\dagger_i$ and  $b^\dagger_i$,
\begin{equation}
\begin{array}{cc}
     A_n = (p_s)^{\frac{1}{2n}}(I_n + \alpha J_n) & B_n = (p_s)^{\frac{1}{2n}}I_n. 
\end{array}    
\end{equation}
For example for $n=3$ and a simple permutation,
$$ A_3 = (p_s)^{\frac{1}{6}}\begin{pmatrix} 1&\alpha&0\\0&1&\alpha\\\alpha&0&1\end{pmatrix}.$$
We can embed these transformations into a unitary matrix defined as 
\begin{equation}
     U_n = \begin{pmatrix}
    A_n & 0 & * \\
    0   &  B_n & *\\
    *  &   *   &*
\end{pmatrix},
\end{equation}
where $p_s$ is the maximum value allowed for the unitary embedding, and the vacuum auxiliary modes are noted $\mathbf{c}$.
Consider an input state in the computational basis represented with the bit string $x$ :
$$\ket{x}_L = \prod\limits_{i=1}^n (a^\dagger_i)^{x_i}(b^\dagger_i)^{1-x_i}\vac . $$

If $\ket{x} =\ket{1,\ldots, 1}_L =\prod\limits_{i=1}^n  a_i^\dagger \vac$, then there are only two terms after the unitary transformation that stays in the computational space.
Either every photon stay on the same mode, either they all cycle to the next modes with coefficient $\alpha$.
Hence, we can introduce a polynomial $ P_{x,n}(\mathbf{a},\mathbf{b},\mathbf{c})$ that has no monomial encoding a n-qubit state, such that :
$$\mathcal{U}\ket{x}_L = \sqrt{p_s}(1+\alpha^n)\ket{x} + P_{x,n}(\mathbf{a},\mathbf{b},\mathbf{c}).$$
The terms of $ P_{x,n}(\mathbf{a},\mathbf{b},\mathbf{c})$ correspond to cases where some of the pairs of modes $a_i^\dagger$ and $b_i^\dagger$ are empty.

Else, if there is at least one logical $0$ in $x$, then the second term with $\alpha^n\prod a_i^\dagger$ where every photon cycles to the next one is not possible as there is at least a missing $a_i^\dagger$. 
Then we have that 
$$\mathcal{U}\ket{x}_L = \sqrt{p_s}\ket{x}_L + P_{x,n}(\mathbf{a},\mathbf{b},\mathbf{c})\vac. $$
Hence this unitary effectively performs a $C^{n-1}Z({\varphi})$ gate with post-selection.

\subsection{Success probability}
Next we give the optimal success probability of this scheme which is the maximum value $p_s$ can be so that $U_n$ stays unitary.
With the unitary extension, we have that $p_s^{\frac{1}{2n}} =1/\sigma_{max}$, where $\sigma_{max}$ is the maximum  singular value of $I_n$ and $I_n+\alpha J_n$.
The latter is circulant and normal.
Its eigenvalues are $\{\lambda_k = 1+\alpha w^k$ for $0\leq k <n\}$, where $w = e^{i\frac{2\pi}{n}}$.
Its singular values are $\sigma_k = |\lambda_k|$:

 \begin{equation}
     \sigma_{max} =  \underset{0\leq k <n}{\text{max}}(|1+\alpha w^k|).
 \end{equation}

$\alpha$ is of the form $(2\text{sin}(\frac{\varphi}{2}))^{\frac{1}{n}}e^{i\frac{\varphi+\pi}{2n}}w^j$ for some $\{0\leq j <n\}$, hence the $w^j$ is absorbed into $w^k$ and 
$$\begin{array}{cccl}
  &\sigma_{max} & = &  \underset{0\leq k <n}{\text{max}}(|1+(2\text{sin}(\frac{\varphi}{2}))^{\frac{1}{n}}e^{i\frac{\varphi+\pi}{2n}}w^k|).\\
 \end{array}$$ 

Hence, the maximum success probability is 
$$\begin{array}{ccc}
p_s &  = & \left( \frac{1}{\underset{0\leq k <n}{\text{max}}(|1+(2\text{sin}(\frac{\varphi}{2}))^{\frac{1}{n}}e^{i\frac{\varphi+\pi}{2n}}w^k|)}\right)^{2n}\\
    & \geq  & \left| \frac{1}{1+(2\text{sin}(\frac{\varphi}{2}))^{\frac{1}{n}}}\right|^{2n}.\\
\end{array}$$ 

\begin{remark}
   The submatrix $B_n$ can be changed for another matrix with the same shape as $A_n$, possibly with a different phase to apply.
   In that case, the optical gate will add a phase to the terms with only ones or only zeros. 
   The success probability will be the lowest among the two phases. 
\end{remark}
\begin{remark}
   This construction can be similarly done for $ n$ qudits instead of qubits. 
   In that case, it will add a phase to one of the terms $\ket{x}$ where $x\in \{0,\ldots,d\mathsmaller{-}1\}^n$, and leave the others unchanged.
\end{remark}
\section{Discussion}

Using a convenient matrix representation, we derived simple rules to determine whether a two-photon state can be prepared using static linear optics.
When the preparation is feasible, our constructions permit to retrieve the corresponding interferometer's unitary, and we checked the validity of the circuits we obtained with the linear optical simulation framework \textit{Perceval} \cite{heurtel2023perceval}.
We did not study how the success probability of these circuits compare with the theoretical optima. 
This could be achieved with a deeper linear algebra analysis or with numerical simulations.
We restricted ourselves to state preparation, but the methods presented here can be used to investigate entangling gates, such as heralded 2-qudit controlled gates.
The matrix representation we used is limited to two photons, and going beyond that requires other methods \cite{paesani2021scheme}. 
As the construction we presented for a generalized post-selected multi-photon gate. 
With this construction, we retrieve existing results when $n=2,3$ \cite{kieling2008linear}, but we go beyond the state of the art in the general case \cite{fiuravsek2006linear}.

\section*{Acknowledgments}
We are grateful to Nicolas Heurtel for valuable discussions and his invaluable help in formalizing the results. We would also like to thank Rawad Mezher for his feedback.
This work has been co-funded by the European Commission as part of the EIC accelerator program under the grant agreement 190188855 for SEPOQC project, by the Horizon-CL4 program under the grant agreement 101135288 for EPIQUE project, and by the UFOQO Project financed by the French State as part of France 2030.

\bibliographystyle{IEEEtran}
\bibliography{bib}

\end{document}